\documentclass[a4paper]{article}
\usepackage[utf8]{inputenc}
\usepackage[T1]{fontenc}
\usepackage{lmodern}
\usepackage{amsmath}
\usepackage{amsthm}
\usepackage{amsfonts}
\usepackage{amssymb}
\usepackage{graphicx}
\usepackage{fullpage}
\usepackage{verbatim}
\usepackage{hyperref}
\usepackage{xcolor}
\usepackage[normalem]{ulem}

\newcommand{\NN}{\mathbb{N}}
\newcommand{\PP}{\mathcal{P}}
\newcommand{\Jpi}{J_{\Pi}}

\newtheorem{thm}{Theorem}
\newtheorem{alg}{Algorithm}
\newtheorem{defi}{Definition}

\newtheorem{lem}{Lemma}

\newtheorem*{vois}{Neighbouring function}
\newtheorem{proc}{Procedure}

\DeclareMathOperator{\complete}{\text{\textsc{complete}}}
\DeclareMathOperator{\enum}{\text{\textsc{enum}}}
\DeclareMathOperator{\prox}{\Tilde{\cap}}
\DeclareMathOperator{\neigh}{\text{\textsc{neighbours}}}
\DeclareMathOperator{\bigo}{\mathcal{O}}

\author{Caroline Brosse$^1$, Aurélie Lagoutte$^1$, Vincent Limouzy$^1$, Arnaud Mary$^2$ and Lucas Pastor$^1$}

\title{Efficient enumeration of maximal split subgraphs and induced sub-cographs  and related classes
\thanks{This research was partially financed by the French government IDEX-ISITE 
initiative 16-IDEX-0001 (CAP 20-25). - A. Lagoutte and A. Mary were supported by ANR Project GrR (ANR-18-CE40-0032)
C. Brosse, A. Lagoutte, V. Limouzy and L. Pastor were supported by the ANR project GRALMECO (ANR-21-CE48-0004-01).
}
}

\begin{document}

\maketitle{}

\begin{center}
\date{1. {\sc Limos}, Univ. Clermont-Auvergne,\\ 2. {\sc Lbbe}, Univ. Claude Bernard- Lyon 1\\
\today}
\end{center}

\begin{abstract}
In this paper, we are interested in algorithms that take in input an arbitrary graph $G$, and that enumerate in output all the (inclusion-wise) maximal ``subgraphs'' of $G$ which fulfill a given property $\Pi$.
All over this paper, we study several different properties $\Pi$, and the notion of subgraph under consideration (induced or not) will vary from a result to another. 

More precisely, we present efficient algorithms to list all maximal split subgraphs, maximal induced cographs and maximal threshold graphs  of a given input graph.
All the algorithms presented here run in polynomial delay, and moreover for split graphs it only requires polynomial space.
In order to develop an algorithm for maximal split (edge-)subgraphs, we establish a bijection between the maximal split subgraphs and the maximal stable sets of an auxiliary graph.
For cographs and threshold graphs, the algorithms rely on a framework recently introduced by Conte \& Uno \cite{conte-uno} called \emph{Proximity Search}.
Finally we consider the extension problem, which consists in deciding if there exists a maximal induced subgraph satisfying a property $\Pi$ that contains a set of prescribed vertices and that avoids another set of vertices.
We show that this problem is NP-complete for every non-trivial hereditary property $\Pi$. We extend the hardness result to some specific edge version of the extension problem.
\end{abstract}

%% Intro
\section{Introduction}
\label{sec:intro}
In this paper, we are interested in algorithms that take in input an arbitrary graph $G$, and that enumerate in output all the (inclusion-wise) maximal ``subgraphs'' $G'$ of $G$ that fulfill a given property $\Pi$.
All over this paper, we study several different graph properties $\Pi$, which all are recognisable in polynomial time and hereditary, where a graph property is \emph{hereditary} if for any graph satisfying it, all its induced subgraphs also do so.
The input graph $G$ will never be restricted to a specific class of graphs.
The notion of subgraph under consideration (induced or not) will vary from a result to another.

%We consider the problem of efficiently listing all maximal subgraphs that fulfil a property $\Pi$ of an arbitrary graph.
%More precisely, we are interested in hereditary properties that can be recognised in polynomial time.
In some applications, for example database search \cite{YYH05}, network analysis \cite{GK07}, or bioinformatics \cite{Dam06,Mar15}, the classical combinatorial approach consisting in finding one single optimal solution is not completely relevant, because the most useful answer is the list of all solutions, instead of a single one. This lead to the design of enumerating algorithms. In particular,
 many algorithms have been developed to find specific maximal subgraphs such 
as maximal stable sets \cite{tsukiyama,avis-fukuda}, spanning trees \cite{ReadT75,GabowM78}, 
or maximal matchings \cite{Uno97,Uno01} to name a few (a more exhaustive list can be obtained in Wasa's survey \cite{Wasa16}).  
For all these problems, the algorithms which were designed are very efficient: they run in polynomial delay and usually only require polynomial space.
More recently, authors started to look at more general subgraphs.
For example, Calamoneri \emph{et al.} \cite{CalamoneriGMSS16} considered  maximal chain graphs in a 
bipartite graph, and Conte \emph{et al.} \cite{conte,conte-uno} considered several graph classes such as 
maximal chordal graphs, maximal bipartite graphs, or maximal $k$-degenerate graphs inside an arbitrary graph.

As we are potentially dealing with a large number of objects, the classical notion 
of efficiency of  an algorithm is no longer relevant. It is often the case that the 
number of solutions is exponential in the size of the input, hence it 
is hopeless to look for polynomial-time algorithms.  In the following we will 
express the complexity of the algorithms in terms of the input size and the output 
size. Such an approach is often called \emph{output sensitive} in the literature 
(in opposition to \emph{input sensitive} where only the input size is taken into account).
For that reason, several complexity classes were introduced by Johnson \emph{et. al.} \cite{johnson-yannakakis}
to capture the notion of efficiency for enumeration algorithms. 
The natural translation of polynomial time complexity into the field of enumeration problems
is called output polynomial: an algorithm is said to be \emph{output polynomial} if 
the time necessary to list all the objects depends polynomially on the size of the input and the 
size of the output. However, this notion is not completely satisfying: since we deal with 
a large number of objects, the time to produce the first solution might sometimes 
be exponential in the input size. For this reason, a refinement of the class was introduced: an algorithm is said to be \emph{incremental polynomial time} when the time necessary to produce a new solution depends polynomially on the size of the input and the number of
%\lily{(c'etait "size of" mais c'est plutot "number of" solutions, non ?)} 
solutions which have already been found. 
It is possible to further refine this notion by eliminating  the dependence on the 
number of already found solutions. An algorithm is said to have \emph{polynomial delay}
when the time necessary to produce a new solution is bounded by a polynomial in the size of the input.
%For some enumeration problems, some constant delay algorithms are known, such
%as, for instance, listing all the permutations, or all the spanning subtrees of 
%a graph \cite{Ruskey} \lily{ref Ruskey inconnue ?}.

%% connections with transversals
The enumeration of some maximal induced subgraphs is related to another important problem in 
enumeration, namely the enumeration of minimal transversals of a hypergraph. This problem 
is also known as the dualisation of monotone boolean formulas. The best algorithm 
for this problem is due to Fredman and Khachiyan \cite{FredmanK96} and runs in quasi-polynomial time. 
%
%~
Enumerating the maximal induced subgraphs of a hereditary class of graphs is equivalent 
to finding all the minimal subsets of vertices that intersect all the minimal forbidden 
induced subgraphs contained in the original graph. 
%~

Thanks to the result of Eiter  \emph{et. al.} \cite{EiterG95} 
on the enumeration of minimal transversals of hypergraphs of bounded edge size, 
it is straightforward to obtain an incremental algorithm for the enumeration 
of maximal induced subgraphs lying in $\mathcal{C}$, where $\mathcal{C}$ is a hereditary class of graphs characterised by a finite number of forbidden subgraphs.
%\lily{(la suite de la phrase etait: "whenever the forbidden subgraphs have a constant number of vertices", mais c'est inutile puisqu'on suppose un nombre fini de graphes, non ?}.
Indeed, a minimal transversal of the hypergraph formed by the forbidden subgraphs corresponds 
to a minimal set of vertices, of which the removal from the original graph yields a maximal 
subgraph in the desired class $\mathcal{C}$.
If in addition the considered class $\mathcal{C}$ is monotone (\emph{i.e.} closed under edge removal), 
the maximal subgraphs can be obtained in a similar way.
\\

%% Related Works
The study of such enumeration problems gave rise to efficient algorithmic frameworks such as \emph{Reverse search} introduced 
by Avis \& Fukuda \cite{avis-fukuda}.
Reverse search provides a general strategy to explore the space of solutions.
Another useful technique called \emph{Binary partition} \cite{ReadT75} or sometimes \emph{Flashlight search}  relies on what we call the \emph{extension problem}.
For maximal induced subgraphs, Cohen \emph{et al.} \cite{cohen} devised an interesting framework which manages to link the complexity of enumerating maximal induced subgraphs that fulfil a hereditary property $\Pi$ to the complexity of a so-called \emph{restricted problem} \cite{lawler}, defined as follows.
Given a graph $H$ that fulfils a hereditary property $\Pi$ and an additional vertex $v$ not in $H$ that is connected to 
some vertices in $H$, the restricted problem consists in asking for the list of all the maximal induced subgraphs of $H+v$ that 
fulfil property $\Pi$.
If the restricted problem admits a polynomial number of solutions (and if these solutions can be found 
in polynomial time), then a polynomial-delay and polynomial space enumeration algorithm can be derived for the general problem. 
They also present analogous results for incremental and output polynomial-time algorithms, which we will not develop here.

Recently, Conte \& Uno introduced a new framework called \emph{Proximity Search} \cite{conte-uno}.
This framework is well-suited for enumerating maximal subgraphs, coping with some of the limitations 
of the \emph{restricted problem} approach, but does not always guarantee polynomial space.
\\

%% Contributions

\begin{table}[ht]
    \centering
    \begin{tabular}{|c||c|c|c|}
    \hline
  Prop. $\Pi$:  & Max. induced subgraphs & Max. edge-subgraphs & Min. edge-supergraphs  \\
  being a.... & & {\small also called} min. deletions & {\small also called} min. completions \\\hline
 Split graph & poly. delay, poly. space & poly. delay, poly. space &  poly. delay, poly. space \\
 & Cao \cite{cao2020} & Theorem \ref{thm:split completions}  & Theorem \ref{thm:split completions} \\ 
  &  \textit{via} \emph{restricted problem} & \textit{via} self-complementarity & \textit{via} ad-hoc bijection  \\ 
  & & Section \ref{sec:split} & Section \ref{sec:split}   \\ \hline
  
   Cograph & poly. delay &  &   \\
 & Theorem \ref{thm: induced cographs} & Open  & Open \\ 
  &  \textit{via} \emph{Proxi. Search} &  &  \\
  & Section \ref{sec:cograph} & & \\ \hline
  
%     $P_3$-free graph & poly. delay, poly. space & poly. delay &  poly. delay \\
%  & Cao \cite{cao2020} & Theorem \ref{thm: P3-free deletions}  & see Observation \ref{obs: P3-free completion} \\ 
%   & \textit{via} \emph{restricted problem} & \textit{via} \emph{Proxi. Search} & trivial \\ \hline
  
  Threshold & poly. delay, poly. space & poly. delay & poly. delay \\
  graph & Cao \cite{cao2020} & Theorem \ref{thm: threshold deletions}  & Theorem \ref{thm: threshold deletions} \\ 
  & \textit{via} \emph{restricted problem} & \textit{via} \emph{Proxi. Search} & \textit{via} self-complementarity \\ 
   & & Section \ref{sec:threshold} & Section \ref{sec:threshold} \\\hline
  
%   trivially & poly. delay & poly. delay & \\
%   perfect graph & Cao \cite{cao2020} & Theorem \ref{thm: trivially perfect deletions} &  ? \\
%   & \textit{via} \emph{retaliation-free paths} & \textit{via} \emph{Proxi. Search} &  \\ \hline
  
    \end{tabular}
    \caption{Summary of known and new algorithms that, given an input graph $G$, enumerate all maximal ``subgraphs'' (or minimal supergraphs) of $G$ that fulfil a property $\Pi$, where $\Pi$ describes a subclass of cographs or a subclass of chordal graphs}
    \label{tab:results}
\end{table}

The results of this paper  split into two categories: first, polynomial-delay enumeration algorithms that are summarised in Table \ref{tab:results} (and contextualised in the following of this introduction) ; secondly, negative results concerning the extension problem. 

For all the provided algorithms, we fix a property $\Pi$ and
a notion of ``subgraph'' (resp. ``supergraph'') and
the goal is, given an arbitrary input graph $G$, to enumerate all inclusion-wise maximal ``subgraphs'' (resp. minimal ``supergraphs'') of $G$.
Each property $\Pi$ under consideration is hereditary and recognisable in polynomial time, and moreover corresponds to a very structured class of graphs, all of which being subclasses of perfect graphs. Because of the multiple applications of chordal graphs and cographs, we focused on them and designed algorithm when $\Pi$ describes 
split graphs (subclass of chordal graphs); 
or cographs themselves;
or threshold graphs (subclass of both chordal graphs and cographs).
Although polynomial-delay algorithms are already known for the enumeration of minimal chordal deletions and maximal induced chordal subgraphs \cite{conte-uno}, it would be of great interest to provide such an algorithm for minimal chordal completions, also known as minimal triangulations.

Observe that the three aforementioned properties $\Pi$ are closed under adding isolated vertices.
Hence, given a graph $G$, all inclusion-wise maximal subgraphs fulfilling $\Pi$ must contain all vertices of $G$.
In other words, all the subgraphs under interest can be described by the set of edges that are removed from $G$, hence we will call them \emph{edge-subgraphs}, or \emph{deletions} of $G$.
In the same fashion, when all cliques fulfil property $\Pi$, it is of interest to study the opposite operation: add edges to $G$ in order to get an edge-supergraph of $G$ fulfilling $\Pi$.
This will be called a \emph{$\Pi$-completion} of $G$.
It is straightforward to observe that enumerating all minimal deletions of $G$ that fulfil property $\Pi$ is equivalent to enumerating all minimal completions of $\overline{G}$ that fulfil property $\Pi$.
So, when $\Pi$  is closed under complement, we might choose to study one problem or the other, depending on which seems easier to visualise.

Notice that some properties $\Pi$ considered in this article have
already been studied in the light of maximal induced subgraphs or maximal edge-subgraph: an algorithm to find a single maximal split edge-subgraph of the input graph has been designed in \cite{heggernes}; and algorithms to find a single maximal edge-subgraph that is a cograph have been given in \cite{crespelle,lokshtanov}.
As for enumeration, Cao \cite{cao2020} recently proved that the above-mentioned restricted problem admits a polynomial number of solutions for several graph classes. From this result, he thus obtained polynomial-delay and polynomial space algorithms to enumerate all maximal induced subgraphs that fulfil property $\Pi$, 
for $\Pi$ referring to split graphs and threshold graphs (see Table \ref{tab:results}).

In Section \ref{sec:split}, we focus on the case where $\Pi$ describes split graphs: 
%If we just want to find one maximal split (edge-)subgraph, it is possible to do so with efficient algorithms, like the one 
%given in \cite{heggernes}.
we provide an efficient algorithm to enumerate all minimal split completions 
of a given input graph $G$,
by establishing a bijection between them and the maximal stable sets of 
an auxiliary graph built from $G$.

In Section \ref{sec:cograph}, the desired property $\Pi$ refers to being a cograph.
%Efficient algorithms to find just one maximal (edge-)subgraph which is a cograph have been desgined in \cite{crespelle,lokshtanov}.
At first sight, one may hope to enumerate all maximal induced subgraphs that are cographs by using the restricted problem; however, it turns out that the restricted problem in this case may have an exponential 
number of solutions, hence it cannot be used to derive a polynomial-delay algorithm.
Fortunately, we are able to adapt the \emph{Proximity Search} framework mentioned above, recently introduced 
by Conte \& Uno \cite{conte-uno}.
In order to make it work, we rely on the construction ordering with true and false twins.
The approach developed for maximal induced cographs can be extended to maximal threshold graphs.
% 
% such as $P_3$-free graphs in Section \ref{sec:P3}: we obtain a polynomial-delay algorithm to enumerate all minimal $P_3$-free deletions of the input graph. Since $P_3$-free graphs are not closed under complement, this does not give an enumeration algorithm for minimal $P_3$-free completions. However, it turns out that there is a single minimal $P_3$-free completion for any input graph $G$, which gives a trivial polynomial-time algorithm (see Observation \ref{obs: P3-free completion}).
% With those two algorithms, one can trivially enumerate all minimal completions or deletions in the complement class, namely $(K_2 + K_1)$-free graphs

Following the same outline as for cographs, we obtain in Section \ref{sec:threshold} a polynomial-delay algorithm to enumerate all minimal threshold deletions.

Finally in Section \ref{sec:extension} we consider the extension problem. 
The extension problem is a decision problem that, if solvable in polynomial time, can be derived 
into a simple and classical polynomial-delay and polynomial space algorithm called \emph{Binary partition} 
or \emph{Flashlight search} for the associated enumeration problem.

The maximal induced subgraph extension problem, studied in Subsection \ref{sec:ext-induced} is the following.
Given a graph $G$, a property $\Pi$ and two disjoint subsets of vertices $A$ and $B$ of $G$, 
does there exist a maximal subgraph of $G$ that fulfils property $\Pi$ and contains all the vertices of $A$ and none of $B$?
We prove that the extension problem for maximal induced subgraphs is NP-complete for every 
nontrivial hereditary property $\Pi$ (here, nontrivial means that $\Pi$ is true for infinitely many but not all graphs).
This proves that such a classical strategy cannot be applied to enumerate maximal induced subgraphs.

Lastly in Subsection \ref{sec:ext-edge}, we study the extension problem for maximal edge-subgraphs. We were not able to obtain such a general 
result as the above-mentioned one, however for basic classes that are $P_k$-free graphs and $C_k$-free graphs for $k$ at least $3$, the respective extension problems are both NP-complete.

\section{Notations and definitions}
%{\bf définitions}

For the remainder of this section, let $G$ be a graph. We denote by $V(G)$ the vertex set of $G$ and 
by $E(G)$ its edge set.
When $G=(V,E)$, and if there is no ambiguity, we may use $V$ and $E$ instead of $V(G)$ and $E(G)$.

%% neighbourhood
The \emph{neighbourhood} of a vertex $v$ in $G$ is denoted $N_G(v)$ and is defined as $N_G(v)=\{w\in V(G) \ | \ vw \in E(G)\}$. The \emph{closed neighbourhood} of a vertex $v$ in 
$G$ is denoted by  
$N_G[v]$ and is equal to $N_G(v)\cup\{v\}$. When it is clear from 
the context, the subscript $G$ might be omitted.
Two vertices $x$ and $y$ are called \emph{twins} whenever 
$N(x)\setminus \{y\} = N(y)\setminus \{x\}$. Two twins are called 
\emph{true twins} if there exists an edge between them, or  \emph{false twins} otherwise. 

Given $X,Y$ two subsets of vertices, the \emph{symmetric difference between} $X$ and $Y$, denoted by $X\triangle Y$, is defined as $(X\cup Y)\setminus (X\cap Y)$.

%% Induced subgraph 
If $X\subseteq V$, then $G[X]$ denotes the subgraph of $G$ induced by $X$, that is $(X, E(G)\cap X^2)$ (obtained from $G$ by removing vertices not in $X$ and edges incident to at least one deleted vertex). A graph $H$ is an \emph{induced subgraph} of $G$ if there exists $X\subseteq V$ such that $H=G[X]$. For $X\subseteq V$, we might sometimes write $G\setminus X$ instead of $G[V\setminus X]$.

%A graph $H$ is an \emph{induced subgraph} of a graph $G$ is $H$ can be obtained from 
%$G$ by removing vertices of $G$. It is often the case that we indicate which are 
%the vertices we want to keep; in that case $G[X]$ denotes the graph induced 
%by the set of vertices $X$.
%% Subgraphs
A graph $H$ is a \emph{subgraph} of $G$ if 
$H$ can be obtained from $G$ by 
removing some vertices and some edges of $G$, that is to say if $V(H)\subseteq V(G)$ and $E(H)\subseteq E(G)$. If moreover $V(H)=V(G)$, then $H$ is called an \emph{edge-subgraph}, or else a \emph{deletion}, of $G$. By abuse of language, if $E'\subseteq E(G)$,  we might call $E'$ a deletion of $G$, referring to $(V(G), E')$. %\lily{Vérifier si c'est cohérent avec ce qu'on fait dans les Proximity Search, qu'un des reviewer n'a pas aimé} C'est cohérent.
%% Supergraphs
A graph $H$ is called an \emph{edge-supergraph} of $G$, or a \emph{completion} of $G$, if $V(H)=V(G)$ and $E(G)\subseteq E(H)$. Then $E(H)\setminus E(G)$ is called the set of \emph{fill edges}.

A graph property $\Pi$ is called \emph{hereditary} when it is closed under vertex removal.
In other words, if the property is fulfilled by the graph we consider, then it is fulfilled by all its induced subgraphs.
All the properties considered in this paper are hereditary.
A graph property is called \emph{monotone} when it is closed under both edge removal 
and vertex removal. That is to say, if the property is fulfilled by $G$,
then the property is fulfilled by all subgraphs of $G$.
A property $\Pi$ is called \emph{sandwich monotone} \cite{heggernes2} if, for any two graphs $G$ and $H$ on the same vertex set $V$, that both
fulfil the property $\Pi$, and such that $H$ is an edge-subgraph of $G$, there 
exists a sequence of edges $s=(e_1,e_2,\ldots,e_k)$ of $E(G)\setminus E(H)$
such that: for any $1\leq i \leq k$, the graph $(V, E(H) \cup s_i)$  
fulfils the property $\Pi$, where  $s_i=\{e_1,\ldots,e_i\}$.

A graph $H$ is a \emph{$\Pi$-edge-subgraph} of $G$, also called a \emph{$\Pi$-deletion} of $G$, if it is an edge-subgraph of $G$ and it fulfils property $\Pi$. It is called a \emph{maximal} $\Pi$-edge-subgraph, or a \emph{minimal} $\Pi$-deletion, if it is inclusion-wise maximal, that is to say there is no $\Pi$-edge-subgraph $H'$ of $G$ such that $E(H) \subsetneq E(H')\subseteq E(G)$.
Similarly, a graph $H$ is a \emph{$\Pi$-edge-supergraph} of $G$, also called a \emph{$\Pi$-completion} of $G$, if it is an edge-supergraph and $G$ and it fulfils property $\Pi$. It is called a \emph{minimal} $\Pi$-edge-supergraph of $G$, or \emph{minimal} $\Pi$-completion, if it is inclusion-wise minimal: there is no $\Pi$-supergraph $H'$ of $G$ such that $E(G) \subseteq E(H')\subsetneq E(H)$.
When $\Pi$ refers to being in some known class of graphs, say being a split graph, we will omit the hyphen and simply say: a split edge-subgraph, a split deletion, a minimal split deletion, and so on.

Finally, when dealing with a total ordering on a set, we consider the \emph{lexicographic order} on finite sequences of elements as follows.
Given two distinct finite sequences, consider the first (leftmost) element on which they differ.
The sequence in which this element is smaller (or does not even exist) is the smallest sequence.

%a subgraph $H$ of $G$ (induced or not) is \emph{maximal for $\Pi$} if any $\Pi$-subgraph of $G$ containing $H$ is $H$ itself.
All over this paper, \emph{maximal} (resp. \emph{minimal}) will always refer to inclusion-wise maximal (resp. minimal), and never to maximum (resp. minimum) cardinality.

\section{Split graphs: minimal split completions}
\label{sec:split}

\emph{Split graphs} \cite{foldes-hammer,golumbic-livre} are graphs whose vertex set
%\lily{the vertex set of which?}
can be partitioned into two sets $K$ and $S$,  where $K$ is a clique of $G$ and $S$ is a stable set.
This partition $K+S$ is called a \emph{split partition} of the graph, and is not necessarily unique.
Split graphs can also be characterised as $(C_4,2K_2,C_5)$-free graphs.
This class of graphs is self-complementary, meaning that for any split graph $G$, its complement $\overline{G}$ is also split.

Split graphs have been studied in detail in  \cite{heggernes} in order to find one single minimal split completion of a given graph, and in \cite{cao2020} to enumerate all maximal split induced subgraphs.
Our aim here is to enumerate all minimal split completions of $G$ into a split graph, using a polynomial-delay polynomial-space algorithm.
To do so we establish a link between minimal split completions of a graph and its maximal stable sets. However we will 
show that different maximal stables sets may lead to the same completion. To overcome this problem we first perfom 
a pre-processing step that remove some unnecessary vertices. Once the process is completed, the bijection between 
the two objects is established.

Recall that a minimal split completion of a graph $G$ is an inclusion-wise minimal split supergraph of $G$ on the same vertex set; that is to say a split graph $H$ such that $V(H)=V(G)$, $E(G) \subseteq E(H)$, and there exists no split graph $H'=(V(G), E(H'))$ such that $E(G)\subseteq E(H') \subsetneq E(H)$. Every edge in $E(H)\setminus E(G)$ is called a \emph{fill edge}.

Since the class of split graphs is self-complementary, we can list all maximal split deletions of $G$ simply by computing all minimal split completions of the complement graph $\overline{G}$, and then complementing the solutions.

\paragraph{}

First, note that a split completion of $G$ can always be obtained from a stable set.
Indeed, if $S$ is a stable set of $G$, adding all possible edges between pairs of vertices of $V\setminus S$ makes $V\setminus S$ into a clique, and the resulting graph is a split completion $H$ of $G$.
We will say that $H$ is the split completion \emph{induced} by $S$.
In the sequel we will see how to characterise the stable sets which will produce different minimal split completions.

\begin{lem}\label{th1}
Let $G=(V,E)$ be a graph.
For any minimal split completion $H$ of $G$, there exists a split partition $V=K+S$ of $H$ such that $S$ is a maximal stable set of $G$.
\end{lem}

\begin{proof}
Let $H$ be a minimal split completion of $G$.
The graph $H$ is split so there exists a split partition, and $H=(S+K, E+F)$ (where $F$ is the set of fill edges). % Définition fill edges
If $S$ is not a maximal stable set of $H$, there exists $x\in V$ such that $S\cup \{x\}$ is a stable set of $H$, and of course $K\setminus \{x\}$ is still a clique of $H$.
Consequently, suppose that $S$ is a maximal stable set of $H$.

For contradiction, suppose that $S$ is not a maximal stable set of $G$.
There exists $y$ such that $S\cup \{y\}$ is a stable set of $G$.
As $S\subseteq S\cup\{y\}$, the split completion induced by $S\cup\{y\}$ is a split sub-completion of $H$, supposed to be minimal.
Therefore $S\cup\{y\}$ is a stable set of $H$, that is, $S$ is not a maximal stable set of $H$: contradiction.
\end{proof}

Hence it is only necessary to consider (inclusion-wise) \emph{maximal} stable sets when looking for \emph{minimal} split completions.

In the following, we identify two kinds of vertices which can cause a solution to be non-minimal or produced twice.
They are true twins and \emph{redundant vertices}, the definition of which relies on nested neighbourhoods.

\begin{defi}[Redundant vertices]
Let $G=(V,E)$ be a graph.
A vertex $v\in V$ is \emph{redundant} if there exists $u\in V$ such that $N[u]\subsetneq N[v]$.
\end{defi}

Remark that, because of the strict inclusion between the closed neighbourhoods, for every redundant vertex $v$ there exists a vertex $u$ which is not redundant such that $N[u]\subsetneq N[v]$.
The set $N[u]$ is inclusion-wise minimal.
The interest of identifying redundant vertices is given by the following observation (a stronger version of this will be proved in Lemma \ref{lem: min completion iff}): if $S$ is a maximal stable set of $G$ not inducing a minimal split completion, then $S$ contains a redundant vertex.
However, the converse is not necessarily true, as shown on Figure \ref{fig: cliquestable}: this graph admits a maximal stable set $S=\{2, 5\}$ containing redundant vertices ($2$, $4$, and $5$ are redundant vertices), still inducing a minimal split completion.
Hence, we need to push further to identify which redundant vertices should be avoided in a stable set $S$ in order for the induced completion to be minimal.

\begin{figure}[!ht]
	\centering
	\includegraphics{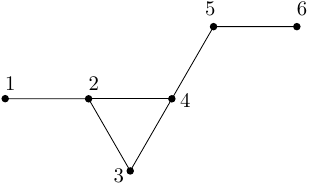}
	\caption{In this graph, the vertices $2$ and $5$ are redundant but the completion induced by the stable set $\{2, 5\}$ is minimal.}
	\label{fig: cliquestable}
\end{figure}

The right characterisation of minimal split completions is in fact the following:

\begin{lem}\label{lem: min completion iff}
Let $G$ be a graph, let $I$ be a maximal stable set of $G$.
The split completion of $G$ induced by $I$ is minimal if and only if $I$ does not contain any redundant vertex $x$ such that $V\setminus N(x)$ is a stable set of $G$.
\end{lem}

In order to prove Lemma \ref{lem: min completion iff}, we need the following lemma which establishes a link between the split partitions of a graph and those of its subgraphs.
In particular it can be used to compare two different split partitions of a split graph.
This result will be invoked when characterising non-minimal split completions.

\begin{lem}[{\cite[Observation 4]{heggernes}} \& Corollary]\label{th3}
Let $G=(V,E)$ and $G'=(V,E')$ be two split graphs such that $E\subseteq E'$, and let $V=I+K$ and $V=I'+K'$ be split partitions of $G$ and $G'$ respectively.
Then the following two inequalities hold:
\begin{itemize}
	\item $\vert K'\cap K\vert \geq \vert K\vert -1$;
	\item $\vert I'\cup I\vert \leq \vert I\vert +1$.
\end{itemize}
\end{lem}

\begin{lem}\label{th4}
Let $G=(V,E)$ be a graph.
Let $I$ be a maximal stable set of $G$.
If the split completion of $G$ induced by $I$ is not minimal, there exists a redundant vertex $w\in I$ such that $N(w)=V\setminus I$.
\end{lem}

\begin{proof}
Let $I$ be a maximal stable set of $G$, let $H=(I+C, E+F)$ be the split completion of $G$ induced by $I$.
Assume that the completion $H$ is not minimal.
Then there exists a nonempty set of edges $A\subseteq F$ such that $H-A$ is a minimal split completion of $G$.
By Lemma \ref{th1} there exists a maximal stable set $S$ of $G$ such that $H-A = (S+K,E+F-A)$ is the split completion of $G$ induced by $S$.

We have the inclusions $G\subseteq H-A\subseteq H$.
By Lemma \ref{th3}, $\vert S\vert \leq \vert S\cup I\vert \leq \vert S\vert +1$.
Because $S$ and $I$ are both maximal stable sets of $G$ we cannot have $I\subseteq S$, so necessarily $\vert S\cup I\vert = \vert S\vert +1$.
Then, $$\vert S\vert +1 = \vert S\cup I\vert = \vert S\vert +\vert I\setminus S\vert.$$
Therefore $\vert I\setminus S\vert =1$, and we also have $\vert S\setminus I\vert >0$ because $S$ is maximal.
Consequently there exist $w\in V$ and $U\subseteq V$ such that $I = T\cup\{w\}$ and $S=T\cup U$, where $U$ and $T=S\cap I$ are disjoint stable sets of $G$.
Figure \ref{redondant} illustrates this situation.

\begin{figure}[!ht]
	\centering
	\includegraphics[scale=0.8]{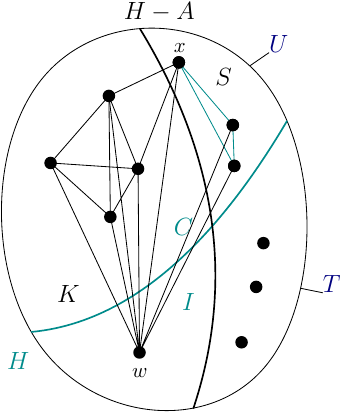}
	\caption{A partition of $V$ for the split completion $H$ and another split partition for $H-A$. The vertex $w$ is universal to $C$. Edges inside $S\cap C$ belong to $H$ but not to $H-A$.}
	\label{redondant}
\end{figure}

All fill edges of $H-A$ are also fill edges of $H$.
Because $w$ is in $I$, there is no fill edge incident to $w$ in $H$, so there is also no fill edge incident to $w$ in $H-A$.
As $w\in K$ in $H-A$, this implies $K\cap C\subseteq N(w)\subseteq (K\cap C)\cup U = C$.

Let $x\in U$.
As $I$ is a maximal stable set of $G$, the vertex $x$ has a neighbour in $I$.
This neighbour is not in $S$ because $x\in S$.
Therefore, it is in $I\setminus S=\{w\}$, that is, $xw\in E$.
Besides, $N[x]\subseteq K\cap C\subseteq N[w]$.
This being true for all $x\in U$, we have $N(w)=C=V\setminus I$.

Moreover, the completion $H$ is not minimal, so there exists $y\in U$ incident to a fill edge of $H$.
Hence there exists $z\in C$ such that $yz\notin E$.
Then $z\in N[w]\setminus N[y]$, so $N[y]\subsetneq N[w]$ and $w$ is redundant.
\end{proof}

Lemma \ref{th4} admits a converse, which is the following.

\begin{lem}\label{th5}
Let $G=(V,E)$ be a graph.
Suppose that there exist a redundant vertex $x$ such that $V\setminus N(x)$ is a stable set of $G$.
If $I$ is a maximal stable set of $G$ containing $x$, then the split completion induced by $I$ is not minimal.
\end{lem}

\begin{proof}
Let $I$ be a maximal stable set of $G$ containing $x$. 
Necessarily $I=V\setminus N(x)$ because $V\setminus N(x)$ is a stable set.
Let $H:=(I+V\setminus I,E+F)$ the split completion induced by $I$.
Since $x$ is redundant, there exists $y\in V$ such that $N[y]\subsetneq N[x]$.
The vertex $y$ is a neighbour of $x$ so $y\in V\setminus I$.
Therefore $y$ is adjacent to every vertex of $V\setminus I$ in $H$.
But $N[y]\subsetneq N[x]$, so there exists $z\in V\setminus I$ such that $yz\notin E$.
In this case, $H-\{yz\in F \mid z\in N(x)\}$ is a proper sub-split completion of $H$.
Hence $H$ is not minimal.
\end{proof}

We can now get the exact characterisation announced in Lemma \ref{lem: min completion iff}.

\begin{proof}[Proof of Lemma \ref{lem: min completion iff}]
Lemma \ref{th4} gives one direction of implication, and Lemma \ref{th5} gives the other (by noticing that $N(w)=V\setminus I$ rewrites to $I=V\setminus N(w)$).
\end{proof}

In the light of Lemma \ref{lem: min completion iff}, the idea is now to remove from $G$ the ``bad'' redundant vertices identified in the lemma,  and then to only enumerate split partitions induced by maximal stable sets of this auxiliary graph. But in order to do so, we need to ensure that maximal stable sets in this smaller graph are also maximal in $G$.

\begin{lem}\label{lem: max stable set when removing redundant vertices}
Let $G=(V,E)$ be a graph and let $R$ be its set of redundant vertices.
For all $R'\subseteq R$, the maximal stable sets of $G\setminus R'$ are maximal stable sets of $G$.
\end{lem}

\begin{proof}
Let $R'\subseteq R$. Recall that $R=\{v\in V\mid \exists u\in V,~N[u]\subsetneq N[v]\}$.
All maximal stable sets of $G\setminus R'$ are stable sets of $G$.
We will see that they are maximal.

Let $S$ be a maximal stable set of $G\setminus R'$.
For contradiction, assume that $S$ is not a maximal stable set of $G$.
Hence there exists $v\in V$ such that $S\cup\{v\}$ is a stable set of $G$, and $v\in R'$ because otherwise $S$ would not be maximal in $G\setminus R'$.
The vertex $v$ is redundant so there exists $u\in V\setminus R$ such that $N[u]\subsetneq N[v]$.
This way, $S\cup\{u\}$ is a stable set of $G$ and a stable set of $G\setminus R'$: contradiction.
\end{proof}

Nevertheless, we are not yet finished. Given a graph $G$, we know that removing some identified redundant vertices from $G$ provides a graph in which every maximal stable set gives a minimal split completion of $G$, and doing so we guarantee to actually get every minimal spit completion of $G$.
However, nothing prevents us from finding the same solution several times.
This phenomenon is illustrated in Figure \ref{meme}, where the stable sets are the sets of circled vertices, and the completion consists only in adding the middle horizontal edge.
To solve this issue, we will characterise stable sets which give the same minimal completion, among the ones satisfying conditions of Lemma \ref{lem: min completion iff}.
This time, the vertices at which we are looking are \emph{true twins}.
This leads us to the following: two distinct maximal stable sets induce the same split completion if and only if their symmetric difference is a pair of true twins. More precisely:

\begin{figure}[!ht]
	\centering
	\includegraphics[page=1]{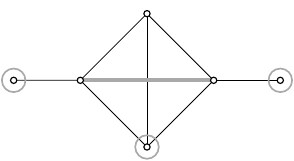}
	\hspace{1cm}
	\includegraphics[page=2]{meme.pdf}
	\caption{The two stable sets represented by circled vertices give the same minimal split completion.}
	\label{meme}
\end{figure}

\begin{lem}\label{lem: kill twins to avoid same split completion}
Let $G$ be a graph, let $S_1$ and $S_2$ be two distinct maximal stable sets of $G$.
Denote by $H_1$ and $H_2$ the split completions of $G$ induced by $S_1$ and $S_2$ respectively.
Define $S:=S_1\cap S_2$.
Then $S_1$ and $S_2$ induce the same split completion (that is to say, $E(H_1)=E(H_2)$) if and only if there exists a pair of true twins $x_1,~x_2\in V$ such that $S_1=S\cup\{x_1\}$ and $S_2=S\cup\{x_2\}$, with $N_G[x_1]=N_G[x_2]=V\setminus S$.
\end{lem}

\begin{proof}
Let $K_1:=V\setminus S_1$ and $K_2:=V\setminus S_2$ be cliques in $H_1$ and $H_2$, respectively.

Assume $S_1$ and $S_2$ induce the same split completion, that is to say $E(H_1)=E(H_2)$.
By Lemma \ref{th3}, we have $\vert S_1\cup S_2\vert = \vert S_1\vert +1$ and $\vert S_1\cup S_2\vert = \vert S_2\vert +1$.
Then $\vert S_1 =\vert S_2$ and there exists $x_1\in S_1\setminus S_2$ and $x_2\in S_2\setminus S_1$ such that $S_1=S\cup \{x_1\}$ and $S_2=S\cup \{x_2\}$.

There is no fill edge incident to $x_1$ in $H_1$.
But because $x_1\in K_2$, we have $K_1\cap K_2 \subseteq N_G(x_1) \subseteq (K_1\cap K_2)\cup\{x_2\}$.
Moreover $x_1 x_2\in E(G)$, otherwise $S_1\cup S_2$ would be a stable set of $G$ containing $S_1$, assumed to be maximal.
This implies that $N[x_1]=V\setminus S$.
By symmetry of the roles of $x_1$ and $x_2$, it gives $N[x_1]=N[x_2]=V\setminus S$, and in particular $x_1$ and $x_2$ are true twins.

Conversely, as $N[x_1]=N[x_2]=V\setminus S$, there is no fill edge incident to $x_1$ in $H_2$, and there is no fill edge incident to $x_2$ in $H_1$.
Consequently, the only fill edges in $H_2$ and in $H_1$ are between vertices of $V\setminus S = K_1\cap K_2$.
Since $H_1$ and $H_2$ are both split, they have the same set of fill edges, so $E(H_1)=E(H_2)$.
\end{proof}

Inspired by Lemmas \ref{lem: min completion iff} and \ref{lem: kill twins to avoid same split completion}, we can then construct from every graph $G$ an auxiliary graph $f(G)$. Let $B=\{x\in V(G) \ | \ x \text{ is redundant and } V\setminus N(x)  \text{ is a stable set} \}$.
First remove $B$ from $G$, and observe that $G[V\setminus B]$ had no extra pair of true twins compared to $G$ itself.
Then from $G[V\setminus B]$, remove all but one vertex from each set of pairwise true twins, resulting in the graph $f(G)$ having no pair of true twins.
The auxiliary graph $f(G)$ is not unique because of the choice between pair of true twins, but is does not matter for the following.

\begin{thm}\label{thm: bijection split}
Let $G$ be a graph.
There exists a bijection between the set of all minimal split completions of $G$ and the set of all maximal stable sets of any auxiliary graph $f(G)$.
%Moreover, the auxiliary graph $G'$ depends only on $G$, and is therefore unique. \lily{On enleve cette phrase ?}
\end{thm}

\begin{proof}
Let $B$ as in the definition of $f(G)$. 
Observe that $G[V\setminus B]$ cannot contain a redundant vertex $x$ of which the non-neighbourhood in $G[V\setminus B]$ is a stable set,
%\lily{est-ce qu'on prend le temps d'expliquer ? Ce n'est pas si trivial qu'on n'en cree pas de nouveau en enlevant B et des true twins, il faut se poser 1 min}
and the same applies to $f(G)$. Moreover, a set $I\subseteq V(f(G))$ is a maximal stable set of $f(G)$ if and only if it is a maximal stable set of $G[V\setminus B]$ included in $V(f(G))$ (because adding a true twin to an existing vertex in $f(G)$ preserves the maximality of $I$).
So, by combining Lemmas \ref{lem: max stable set when removing redundant vertices} and \ref{lem: min completion iff}, a set $I\subseteq V(f(G))$ is a maximal stable set of $f(G)$ if and only if the split completion induced by $I$ is a minimal split completion of $G$.
This, combined with the non-existence of true twins in $f(G)$ and Lemma \ref{lem: kill twins to avoid same split completion}, ensures the bijection.
\end{proof}

The following procedure produces the auxiliary graph $f(G)$ mentioned in Theorem \ref{thm: bijection split}, containing no redundant or true twin vertex $v$ such that $V\setminus N(v)$ is a stable set.

\begin{proc}\label{alg1}
Construction of the auxiliary graph.

\textbf{Input}: a graph $G$

\textbf{Output}: the auxiliary graph $f(G)$
\begin{enumerate}
%    \item For all vertex $x\in V$, compute $N[x]$.\label{etape1}
    \item Mark redundant and true twin vertices:\label{etape2}
    \begin{itemize}
        \item Begin with all vertices unmarked.
        \item For each unmarked $x\in V$, search the list of unmarked $y\in N(x)$.
        If $N[x]\subseteq N[y]$, mark $y$.
    \end{itemize}
    \item For each marked vertex $x$, \label{etape2.5}
    \begin{itemize}
    		\item[] If $V\setminus N[x]$ is a stable set of $G$, then remove $x$.
    \end{itemize}
\end{enumerate}
\end{proc}

During Step \ref{etape2}, every redundant vertex will be marked, as well as all but one vertex of each set of pairwise true twin vertices.

\paragraph{}
Theorem \ref{thm: bijection split} ensures that we can enumerate exactly once every minimal split completion of an input graph $G$ by running Procedure \ref{alg1} followed by an algorithm that enumerates all maximal stable sets of $f(G)$.
Since Procedure \ref{alg1} is polynomial in time and in space, the complexity of this algorithm only depends on the complexity of the algorithm used for the enumeration of maximal stable sets.
We are looking for such an algorithm running in polynomial delay and polynomial space.
For example, the algorithm given in \cite{tsukiyama} fulfils these complexity requirements and allows to state the following.

\begin{thm}\label{thm:split completions}
Minimal split completions of a graph can be enumerated in polynomial delay and polynomial space.
\end{thm}

According to Lemma \ref{th1}, the number of maximal split completions of a graph $G$ is at most the number of its maximal stable sets.
It is proven in \cite{moon-moser1965} that this number can be as large as $3^{n/3}$, reached for a graph built from $n/3$ disjoint triangles as in Figure \ref{triangles}.
\begin{figure}[!ht]
	\centering
	\includegraphics[scale=0.8]{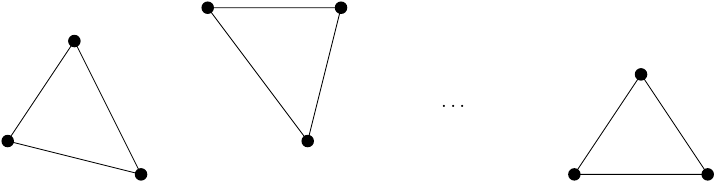}
	\caption{A graph built from $n/3$ disjoint triangles.}
	\label{triangles}
\end{figure}
Lemma \ref{lem: min completion iff} states that there is a bijection between the minimal split completions and the maximal stable sets of a graph $G$ with no redundant vertex $u$ such that $V(G)\setminus N(u)$ is a stable set.
The graph presented in Figure \ref{triangles} verifies this property.

\begin{thm}
The upper bound $3^{n/3}$ on the number of minimal split completions is tight.
\end{thm}

\section{Cographs}
\label{sec:cograph}

Cographs (see \cite{corneil}) can be constructed recursively using disjoint union and complement.
Their characterisations are various and numerous: also known as $P_4$-free graphs, cographs benefit from a unique tree representation, the \emph{cotree}.
Because of this, cographs can be recognised in linear time \cite{corneil-algo}.

An \emph{induced sub-cograph} of $G$ is an induced subgraph that is also a cograph; and it is moreover maximal if it is inclusion-wise maximal among all induced sub-cographs of $G$. To the best of our knowledge, no polynomial-delay algorithm is yet known for the enumeration of maximal induced sub-cographs of a graph.

% \subsection{Maximal induced cographs}

In the case of cographs, the \emph{restricted problem} defined in \cite{cohen,lawler} cannot be solved in polynomial time.
Indeed, the graph $G$ of Figure \ref{pbrestreint} verifies that $G-v$ is a cograph, whereas removing one vertex among $\{x_i,y_i\}$ for all $1\leq i\leq k$ produces a maximal induced sub-cograph of $G$.
This way, we obtain $2^k$ maximal induced sub-cographs of $G$, which implies that the restricted problem has an exponential number of solutions in this case.
Hence, the approach presented in \cite{cohen} cannot be used to produce a polynomial-delay algorithm for the enumeration of maximal induced sub-cographs.

\begin{figure}[!ht]
	\centering
	\includegraphics[page=1]{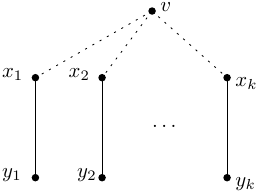}
	\caption{The graph $G-v$ is a cograph, and $G$ has an exponential number of maximal induced sub-cographs.}
	\label{pbrestreint}
\end{figure}

Nevertheless, adding one vertex to a cograph produces a polynomial number of induced $P_4$.
In this case, the results presented in \cite{khachiyan} can be used to obtain an incremental polynomial time algorithm for the enumeration of maximal induced sub-cographs.
We show here that this time complexity can be improved to polynomial delay with the \emph{Proximity Search} introduced in \cite{conte-uno}.

Proximity Search will also be used in Section \ref{sec:threshold} to enumerate all minimal threshold deletions of a given graph.

\paragraph{}
As many enumeration paradigms (for example \cite{avis-fukuda,cohen}), Proximity Search relies on an organised walk among all solutions.
The general idea is to go from a solution to another using a problem-specific transition operation, which cannot be reversed in general, thus implicitly building a directed \emph{solution graph} (the vertices of which are all wanted solutions of the problem).
Showing that the solution graph is strongly connected, by showing we can increase the \emph{proximity} to a target solution at each step, ensures that the search will go trough every solution: all of them will be enumerated.

There are three major requirements to apply Proximity Search.
\begin{itemize}
	\item A notion of proximity $\prox$ between two solutions.
	This proximity does not need to induce a distance because it is not symmetric in general, but it must satisfy $\vert S\prox S'\vert = \vert S'\vert \Rightarrow S = S'$.
	\item A function $\neigh$, computing in polynomial time all (out-)neighbours of a given solution.
	In particular, every solution has a polynomial number of neighbours.
	\item A function $\complete$, adding elements to a given solution, in order to produce a maximal one.
	The use of this function is mostly hidden in the computation of $\neigh$.
	Therefore, the function $\complete$ needs to be computable in polynomial time.
\end{itemize}

The general \emph{Proximity Search} algorithm is the following.

\begin{alg}\label{alg2}
\emph{Proximity Search \cite{conte-uno}}

\textbf{Input}: a graph $G$, an enumeration problem $\PP$

\textbf{Output}: all (maximal) solutions of $\PP$ in $G$ %\lily{On peut bien virer le "maximal" ici, non ? Il est plutot dans la definition de notre $\PP$?}
\begin{enumerate}
    \item $\mathcal{S}\leftarrow \varnothing$
    \item Let $S$ be an arbitrary solution of $\PP$.
    \item Run $\enum(S,\mathcal{S})$.
%    \item Return $\mathcal{S}$. \lily{Ce return n'est-il pas un peu dangeureux lorsqu'on parle de delai polynomial ? Si on return tout a la fin, y'a qu'un seul delai et il est exponentiel}
\end{enumerate}
~
\begin{enumerate}
    \item[] Function $\enum(S,\mathcal{S})$:
    \item $\mathcal{S} \leftarrow \mathcal{S}\cup\{S\}$\\
    {\tt Output $S$ if recursion depth is even.}
    \item For each $S'\in\neigh(S)$ do: \label{etape2b}\\
    \hspace*{1cm} if $S'\notin \mathcal{S}$, then $\enum(S',\mathcal{S})$\\
    {\tt Output $S$ if recursion depth is odd.}
\end{enumerate}
\end{alg}

For all $S$, $\neigh(S)$ is the set of solutions that are neighbours of $S$.
As discussed earlier, these solutions are obtained \textit{via} a transition operation which is specific to the problem.
To obtain a polynomial-delay algorithm, we need to be able to compute all neighbours of $S$ in time polynomial in $n$ for all $S$.
This ensures that the function $\enum$ called in the algorithm will be polynomial.

Note that it is also necessary to be able to find one first solution in polynomial time to have a polynomial-delay algorithm.
Plus, if exponential space is allowed, it is possible to test in linear time if $S'\in \mathcal{S}$, using a dictionary in which insertion and verification (Step \ref{etape2b}) are possible in time polynomial in the size of $S'$.
With all these elements, Algorithm \ref{alg2} runs in polynomial delay.

\paragraph{}
All cographs can be constructed from an isolated vertex by adding true or false twins, as shown in \cite{corneil}.
From such a construction, we obtain an order $v_1\ldots v_n$ on the vertices of a cograph $G$, such that for all $j\in \{1,\ldots,n\}$, $v_j$ has at least one (true or false) twin in $G[v_1,\ldots,v_j]$.
Every maximal induced sub-cograph can then be identified with an ordered set of vertices; this will be done in the sequel.

To define if two solutions are ``close'' (high proximity) or not, it seems natural to compare their twin construction orderings, but these orderings are not unique; the first thing to do is then to make them canonical.

Remark that, as the class of cographs is hereditary, the greedy $\complete$ function which adds one vertex at a time while the result is a cograph can be used, and it runs in polynomial time (because the recognition algorithm for cographs does).
Hence, to use Proximity Search on the class of cographs, we need:
\begin{itemize}
	\item a \emph{canonical ordering} of $V(G')$ for each induced sub-cograph $G'$ of $G$, based on an arbitrary ordering of the vertices of $G$; this canonical ordering is then used to define the \emph{proximity} between two maximal induced sub-cographs;
	\item a \emph{neighbouring function}, enabling to go from one maximal induced sub-cograph $S$ to some other maximal induced sub-cographs, among which one will provably be closer to a target solution, by containing a well-chosen vertex $x\notin S$.
\end{itemize}

Let $G$ be a graph for which we want to enumerate all maximal induced sub-cographs.
Let $u_1,\ldots,u_n$ be an arbitrary order on the vertices of $G$; $u_i$ is said to be smaller than $u_j$ if $i<j$.% \lily{On peut juste dire que $u_i$ is smaller than $u_j$ if $i<j$, et ensuite definir l'ordre lexico uniquement sur les construction ordering ? cf remarque du reviewer}.

\begin{defi}[Canonical ordering for induced cographs] \label{ordre-cographes}
The canonical ordering\footnote{The canonical ordering is not actually used in the procedure but it can be found in polynomial time from the cotree associated to $S$.} of an induced sub-cograph $S$ of $G$ is the lexicographically smallest twin construction ordering of $S$, with respect to $u_1,\ldots,u_n$.
In particular, it is a sequence $v_1\ldots v_k$ of vertices of $G$ such that $v_1$ is the smallest vertex of $G$ belonging to $S$.
\end{defi}

The canonical ordering defined above enables us to define a notion of proximity between two solutions $S$ and $S'$ as follows.

\begin{defi}[Proximity for maximal induced sub-cographs] \label{proximite-cographes}
Let $S$ and $S'$ be two solutions, \textit{i.e.} two maximal induced sub-cographs of $G$.
Denote $k=\vert S'\vert$ and consider $v'_1\ldots v'_k$ the canonical ordering of $S'$.
The \emph{proximity} $S\prox S'$ between $S$ and $S'$ is the longest prefix $v'_1\ldots v'_i$ of the canonical ordering of $S'$ such that $\{v_1', \ldots, v_i'\}\subseteq S$.
In this case, $\vert S\prox S'\vert = i$.
\end{defi}

Note that this notion of proximity is not symmetric and the only solution $S$ maximising $S\prox S'$ is $S'$ itself.

The next step is to define a neighbouring function, called $\neigh$, to construct the solution graph.
To define $\neigh$, we use a greedy function $\complete$, following the arbitrary ordering $u_1, \ldots, u_n$, to complete any sub-cograph into a maximal one, as already mentioned. % \lily{Est-ce qu'on ajoute qu'on suit l'ordre arbitraire sur $V(G)$ pour le greedy, de telle sorte a facilement avoir un complete deterministe ?}
%Because the class of cographs is hereditary, this function adds greedily all possible vertices of $G$ until the resulting sub-cograph is maximal.
Recall that for all $S$, the cardinality of $\neigh(S)$ must be polynomial in $n$.
For simplicity of notation, we might sometimes refer to $S$ as a set of vertices instead of a graph (without introducing any ambiguity on the set of edges, since we are only considering induced subgraph of $G$ in this subsection).
%Since recognition is polynomial, the function $\complete$ runs in polynomial time, ensuring that $\neigh$ also does so.

\begin{vois}
Let $S$ be a  maximal induced sub-cograph.
For all $x\notin S$ and all $y\in S$, define
$$S'_{xy}:=\complete((S\setminus(N_G(x)\triangle N_S(y)))\cup\{x,y\})$$

Then let $\neigh(S):=\left\lbrace S'_{xy}\mid x\notin S,~y\in S\right\rbrace$.
Each solution has at most a quadratic number of neighbours.
\end{vois}

The idea behind this definition of neighbouring function is to add vertex $x$ and make sure that $x$ and $y$ become twins (true or false according to their adjacency) in $S'_{xy}$.
The neighbouring function for induced sub-cographs is illustrated in Figure \ref{completesommets}.
Circled vertices represent the elements of $(N_G(x)\triangle N_S(y))$, %\lily{$N_G(x)\triangle N_G(y)$ (si modif def acceptee, sinon penser a backtrack sur figure)}
two of which are in $S$, so are removed so that $x$ and $y$ become twins in $(N_G(x)\triangle N_S(y)))\cup\{x,y\}$.

\begin{figure}[!ht]
	\centering
	\includegraphics[scale=.9]{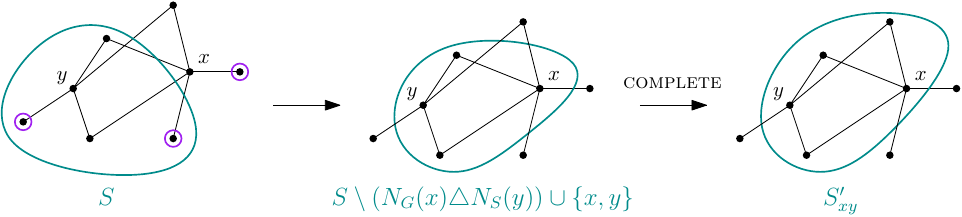}
	\caption{Illustration of the neighbouring function for maximal induced sub-cographs.}
	\label{completesommets}
\end{figure}

\begin{lem}\label{le1}
For every maximal induced sub-cograph $S$, and for every $x\notin S$, $y\in S$, the graph $S'_{xy}$ is a maximal induced sub-cograph of the input graph $G$.
\end{lem}

\begin{proof}
Let $\hat{S}:=(S\setminus(N_G(x)\triangle N_S(y)))\cup\{x,y\}$ and suppose that $\hat{S}$ is not a cograph.
Then there is an induced $P_4$ in $\hat{S}$.
Remark that $\hat{S}\setminus\{x\}$ is a cograph, as an induced subgraph of the cograph $S$.
Hence necessarily $x$ participates to this newly-induced $P_4$, and $y$ does not because it is twin with $x$.
But substituting $x$ for $y$ produces an induced $P_4$ in $\hat{S}\setminus\{x\}$, which is excluded.
Therefore $\hat{S}$ is a cograph.

The function $\complete$, when given an induced sub-cograph of $G$, returns a maximal one.
Consequently, $S'_{xy}$ is indeed a maximal induced sub-cograph of $G$.
\end{proof}

The previous lemma guarantees that the neighbouring function is well-defined and that all neighbours of a solution are also solutions.
Now that we have a solution graph, it remains to show that this solution graph is strongly connected; in other words, that every solution can be reached from any other one by increasing proximity at each step.

\begin{lem}\label{le2}
Let $S$ and $S^*$ be two distinct maximal induced sub-cographs of a given graph $G$.
There exist $x\notin S$ and $y\in S$ such that $S'_{xy}$ is a maximal induced sub-cograph of $G$, satisfying $\vert S'_{xy}\prox S^*\vert > \vert S\prox S^*\vert$.
\end{lem}

\begin{proof}
Denote by $v^*_1,\ldots,v^*_l$ the canonical ordering of $S^*$.

If $\vert S\prox S^*\vert =0$, then $v^*_1\notin S$.
Let $y\in S$, then $v_1^*\in S'_{v_1^*y}$ and $\vert S'_{v^*_1 y}\prox S^*\vert \geq 1 > \vert S\prox S^*\vert$.

Else, let $k$ be the smallest index such that $v^*_k\notin S$ (in this case we have $\vert S\prox S^*\vert = k-1$).
There exists $i<k$ such that $v^*_k$ and $v^*_i$ are twins in $S^*[v^*_1,\ldots,v^*_k]$.
By minimality of $k$, $v^*_i$ is a vertex of $S$.

Consider $S'_{v^*_k v^*_i}$.
It is a cograph by Lemma \ref{le1}.
It remains to show that $S'_{v^*_k v^*_i}$ contains $v^*_1,\ldots,v^*_k$.
For contradiction, assume that there exists $j\in\{1,\ldots,k-1\}$ such that $v^*_j \notin S'_{v^*_k v^*_i}$.
Since $\{v^*_1,\ldots,v^*_{k-1}\}\subseteq S\cap S^*$, it implies that $v^*_j\in N_G(v^*_k)\triangle N_S(v^*_i)$.% \lily{(ou plutot si chgt def) that $v^*_j\in S \cap ( N_G(v^*_k)\Delta N_G(v^*_i))$}
But by definition of the canonical ordering of $S^*$, $v^*_i$ and $v^*_k$ are twins in $G[v^*_1,\ldots,v^*_k]$, so no vertex of $\{v^*_1,  \ldots, v^*_k\}$ can be in $N_G(v^*_k)\triangle N_G(v^*_i)$.
%In particular $v^*_j\in N_{S}(v^*_i)$ implies $v^*_j\in N_{S^*}(v^*_i)$, and $v^*_j\in N_{S^*}(v^*_k)$.
%Thus it is not possible to have $v^*_j \in N_S(v^*_i)\setminus N_G(v^*_k)$, so $v^*_j\in N_{S^*}(v^*_k)\setminus N_S(v^*_i)$.
%This last case is also impossible because $v^*_i$ and $v^*_k$ are twins in $S^*[v^*_1,\ldots,v^*_k]$.
% Phrase ci-dessous a garder
Therefore, $S'_{v^*_k v^*_i}$ contains $v^*_1,\ldots,v^*_k$.

Consequently, $\vert S'_{v^*_k v^*_i}\prox S^*\vert \geq k > \vert S\prox S^*\vert$.
\end{proof}

By Lemma \ref{le2}, it is possible to apply Proximity Search on induced cographs.
Algorithm \ref{alg2} can then be used to enumerate all maximal induced sub-cographs in polynomial delay, hence the following.

\begin{thm} \label{thm: induced cographs}
Maximal induced sub-cographs of a graph $G$ can be enumerated in polynomial delay, and exponential space.
\end{thm}

\section{Threshold graphs: minimal threshold deletions}
\label{sec:threshold}

\emph{Threshold graphs} are at the intersection between split graphs and cographs.
Therefore, they are characterised by forbidden subgraphs as the class of $(P_4, C_4, 2K_2)$-free graphs. This class is quite common in the literature (see for example the chapter dedicated to threshold graphs in \cite{golumbic-livre}).
Like the split and cograph properties, the threshold property is hereditary and closed under complement.

As proven in \cite{chvatal1977}, all threshold graphs can be constructed from a single isolated vertex by adding at each step either an isolated vertex or a universal vertex, \textit{i.e.} a vertex adjacent to all other vertices of the graph. %\lily{universal (et universal, plus loin) ne sont pas defini, je crois} 
This construction provides an ordering $v_1,\ldots,v_n$ on the vertices of a threshold graph $G$, such that for all $1\leq i\leq n$, either $v_i$ is isolated in $G[v_1,\ldots,v_i]$, or $v_i$ is universal in $G[v_1,\ldots,v_i]$.
If several orderings are possible, take the lexicographically smallest, according to an arbitrary order on the vertices, to make it canonical.

Minimal threshold deletions can then be enumerated with the help of Proximity Search, as we will show in the following.
Since the class of threshold graphs is stable under complement, it is equivalent to enumerate minimal threshold completions as well.

Recall that a threshold deletion of $G$ is a subgraph $S$ of $G$ such that $S$ is a threshold graph.
A threshold deletion is moreover \emph{minimal} if the subgraph is inclusion-wise maximal among all threshold subgraphs of $G$.
Since threshold graphs are closed under adding isolated vertices, any minimal threshold deletion $S$ satisfies $V(S)=V(G)$.
Hence, there will never be any ambiguity on the vertex set, so $S$ will often be used to denote both the subgraph and its set of edges $E(S)\subseteq E(G)$. 

As in Section \ref{sec:cograph}, in order to use Proximity Search we need to define a canonical ordering of a solution (here, on the edges, not on the vertices), a notion of proximity between two solutions, a \textsc{Complete} function to transform a threshold subgraph into a maximal one,
 and a neighbouring function to go from one solution to some others.

\paragraph{}
Consider the \textbf{edge canonical ordering} induced by the vertex canonical ordering.
That is to say, order the edges of a threshold graph $S$ in a sequence $((v_j v_i)_{i<j})_{j \leq n}$, where $v_1,\ldots,v_n$ is the (lexicographically smallest) threshold construction ordering of $S$.
The \textbf{proximity} between two minimal threshold deletions $S$ and $S^*$ is defined similarly as before as the longest prefix of $S^*$ included in $S$.
Finally, the class of threshold graphs is proven in \cite{heggernes2} to be \emph{sandwich-monotone}.
It is then possible to derive a polynomial $\complete$ function for threshold graphs, trying to add one edge at a time until we have a maximal threshold subgraph.

The suitable neighbouring function for this problem is the following: to transform a solution into another, choose a vertex and try to make it universal.

\begin{vois}
Let $G$ be a graph and let $S$ be a minimal threshold deletion of $G$.
For a vertex $x$ of $G$, we will build a graph $\widetilde{S}_x$ in which $x$ is universal in the construction ordering.
Define $E_G(x) := \{xt\mid t\in N_G(x)\}$, and consider the following sets of edges of $G$:
\begin{enumerate}
	\item $\widetilde{E}(x) := E_G(x)$;
	\item for all $v\in N_G(x)$, $\widetilde{E}(v) := \{vw\mid w\in N_S(v)\}\cup \{vx\}$;
	\item for every other vertex $v$, $\widetilde{E}(v) := \varnothing$.
\end{enumerate}
Let $\widetilde{S}_x := \complete\left( \bigcup\limits_{v\in V} \widetilde{E}(v)\right)$.
Define $\neigh(S) := \{\widetilde{S}_x\mid x\in V\}$.
Each solution has a linear number of neighbours.
\end{vois}

\begin{lem}
If $S$ is a minimal threshold deletion of a graph $G$, then for all $x\in V$, $\widetilde{S}_x$ is a minimal threshold deletion of $G$.
\end{lem}

\begin{proof}
Let $x\in V$.
It suffices to show that $\widetilde{S} := \bigcup\limits_{v\in V} \widetilde{E}(v)$ is a threshold graph.

In $\widetilde{S}$, all non-neighbours of $x$ are isolated vertices.
Besides, adjacencies between vertices of $N_G(x)$ have not been modified from $S$.
Moreover, since $S$ is a threshold graph, $S[N_G(x)]$ is also a threshold graph; the adding of $x$ as a universal vertex in $S[N_G(x)\cup\{x\}]$ implies that $S[N_G(x)\cup\{x\}]$ is also a threshold graph.
Therefore, as the union of a threshold graph and isolated vertices is a threshold graph, $\widetilde{S}$ is a threshold graph.

It is then possible to apply the polynomial $\complete$ function to build $\widetilde{S}_x$, which is a minimal threshold deletion of $G$.
\end{proof}

Once we know that all neighbours of a solution are also solutions, let us prove that the solution graph is strongly connected.

\begin{lem}
Let $S$ and $S^*$ be two minimal threshold deletions of a graph $G$.
There exists $x\in V$ such that $\vert \widetilde{S}_x\prox S^*\vert > \vert S \prox S^*\vert$.
\end{lem}

\begin{proof}
Let $v_1^*,\ldots,v_n^*$ be the threshold construction ordering of $S^*$, and consider the associated edge canonical ordering.

Let $v_k^* v_i^*$ be the first edge of $S^*$ which does not appear in $S$.

By definition, the neighbours of $v_k^*$ in $\widetilde{S}_{v_k^*}$ are all the neighbours of $v_k^*$ in $G$.
Necessarily, we have the inclusion $N_{S^*}(v_k^*) \subseteq N_{\widetilde{S}_{v_k^*}}(v_k^*)$.
Hence all edge incident to $v_k^*$ in $S^*$ is also present in $\widetilde{S}_{v_k^*}$.

Moreover, if an edge $v_j^* v_i^*$ with $j<k$ is present in $S^*$, then it is also present in $S$ by definition of $v_k^* v_i^*$ as the first edge of $S^*$ which does not appear in $S$.
Since $v_k^* v_i^*$ is an edge of $S^*$, it means that $v_k^*$ is universal in $S^*[v_1^*,\ldots,v_k^*]$.
Therefore every edge $v_j^* v_i^*$ with $j<k$ appearing in $S^*$ is an edge joining two neighbours of $v_k^*$ in $G$: it is not removed by the transition operation, so it appears in $\widetilde{S}_{v_k^*}$.

Hence $\widetilde{S}_{v_k^*}$ satisfies $\vert \widetilde{S}_x\prox S^*\vert > \vert S \prox S^*\vert$.
\end{proof}

From these two lemmas, we deduce that minimal threshold deletions are Proximity Searchable, hence the following theorem.

\begin{thm}\label{thm: threshold deletions}
Minimal threshold deletions (resp. completions) of a graph $G$ can be enumerated in polynomial delay, and exponential space.
\end{thm}

\section{Extension problem}
\label{sec:extension}

\subsection{Maximal induced subgraphs}
\label{sec:ext-induced}

We are interested in one classical method to design some enumeration algorithm, called \emph{Binary Partition} \cite{ReadT75} or sometimes \emph{Flashlight search} .
This method relies on what we call the \emph{extension problem}.
Here we focus on the case of the extension problem applied to the enumeration of maximal induced $\Pi$-subgraphs: given a graph $G$, and two disjoint subsets $A$ and $B$ of its vertex set, does there exist a maximal induced $\Pi$-subgraph of $G$ containing $A$ and avoiding $B$?
As discussed in the following, solving this problem can lead to a polynomial-delay and polynomial space algorithm to enumerate, from a given arbitrary input graph $G$, all maximal induced subgraphs of $G$ satisfying property $\Pi$.

\paragraph{}
Recall that a graph property $\Pi$ is {hereditary} if for all $G$ satisfying $\Pi$, all induced subgraphs of $G$ also satisfy $\Pi$.
The hereditary property $\Pi$ is \emph{nontrivial} if it is verified by infinitely many graphs and has at least one obstruction, which is a graph $\Jpi$ not satisfying $\Pi$. % (that is to say, infinitely many graphs do not satisfy $\Pi$)

We define the \emph{extension problem} for a hereditary graph property $\Pi$ as follows.
Given a graph $G$ (not satisfying $\Pi$ in general) and two disjoint subsets $A$ and $B$ of its vertex set $V$, decide if there exists a maximal induced subgraph $H$ of $G$ satisfying $\Pi$, containing $A$ and avoiding $B$.
In other words, we want that $A \subseteq V(H)$ and $B \cap V(H)=\varnothing$.
Note that we cannot remove $B$ from the vertex set before looking for a maximal induced $\Pi$-subgraph: this operation could lead us to a non-maximal induced $\Pi$-subgraph of $G$, which is not what we are looking for.

If it were possible to solve the extension problem for $\Pi$ in polynomial time, then we would be able to derive a polynomial-delay polynomial space algorithm to enumerate all maximal $\Pi$-subgraphs of a given graph $G$.
Indeed, if we denote by $v_1,\ldots ,v_n$ the vertices of $G$, we can start by checking the existence of a solution containing $v_1$, and the existence of a solution excluding $v_1$.
Iterating this process for all vertices of $G$ each time such a solution exists would lead us to find each solution exactly once, without considering the fail cases.
Such a procedure is referred to as Binary Partition or Flashlight Search.

The \emph{node-deletion problem} consists in determining the \emph{minimum} number of nodes which must be deleted from a graph such that the resulting subgraph satisfies property $\Pi$.
It has been shown by Lewis and Yannakakis in \cite{yannakakis} that for all nontrivial hereditary property $\Pi$, the node-deletion problem is NP-complete, provided that testing for $\Pi$ can be performed in polynomial time.
We will show here that the same can be said about the extension problem.

\begin{thm}\label{extension}
The extension problem for any nontrivial hereditary graph property $\Pi$ is NP-hard. If moreover $\Pi$ can be tested in polynomial time, then the extension problem is NP-complete.
%\lily{C'est bien ca que tu voulais dire Caroline ?}
\end{thm}

%From now on, we will assume that property $\Pi$ can be tested in polynomial time, for if it is not the case, the previous result turns out to be a NP-hardness result.

\paragraph{} 
From now on, with a slight abuse of notation we denote by \emph{the class $\Pi$} the class of graphs satisfying property $\Pi$.
First of all, the class $\Pi$ contains all the stable sets or all the cliques.
Indeed, because $\Pi$ is nontrivial there exist graphs in $\Pi$ that have arbitrarily many vertices.
Let $k \in \NN$.
Then by Ramsey's theorem, every graph in $\Pi$ that is large enough (at least the so-called Ramsey number $R(k,k)$) contains either a clique or a stable set of size $k$.
Since $\Pi$ is a hereditary property, it implies that this clique or stable set of size $k$ belongs to $\Pi$.
Up to considering the complement class $\overline{\Pi}$ (which remains nontrivial and hereditary), we can assume that $\Pi$ contains all the stable sets.
%Considering the complement of a solution in $\overline{\Pi}$ gives a solution in $\Pi$.

\paragraph{}
The first step is to define a total ordering on the set of all graphs, essentially the same as the one given in \cite{yannakakis}.
This order will be useful in the proof of Theorem \ref{extension}, to ensure that the considered graph satisfies property $\Pi$.

For any connected graph $G$ and any vertex $c \in V(G)$, let $\lambda^c(G)$ be the non-increasing sequence of the sizes of the connected components of $G-c$.
If $V(G)=\{c\}$, then we take $\lambda^c(G)=(0)$.
Let then $\lambda(G)$ be the minimum for the lexicographic order of the $\lambda^c(G)$, that is, $\lambda(G)=\min\limits_{lex} \left\lbrace \lambda^c(G) \mid c\in V(G) \right\rbrace$.
Let $c(G) \in V(G)$ be the vertex which minimises $\lambda^c(G)$.
If the minimum is reached for several vertices, take any of them as $c(G)$.

For any graph $G$ with connected components $G_1,\ldots,G_l$, let $\mu(G)$ be the non-increasing sequence of the $\lambda(G_i)$ according to the lexicographic order.
For convenience, the connected components $G_1,\ldots,G_l$ of $G$ will now be ordered to have $\mu(G) = (\lambda(G_1),\ldots,\lambda(G_l))$.
The function $\mu$ defines a total and well-founded order (all elements are pairwise comparable, there is no infinite decreasing sequence for $\mu$) on the set of all graphs.

Figure \ref{ordre} illustrates this order on three connected graphs $G_1$, $G_2$, and $G_3$, where the vertex $c$ minimising $\lambda$ is identified (except for $G_2$ where any vertex can be taken).

\begin{figure}[!ht]
	\centering
	\includegraphics{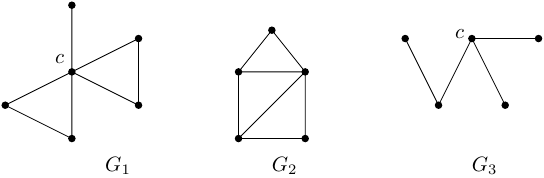}
	\caption{In this example, we have $\mu(G_1)=((2,2,1))$ whereas  $\mu(G_2)=((4))$, and $\mu(G_3)=((2,1,1))$. Therefore, $\mu(G_3) \leq \mu(G_1) \leq \mu(G_2)$, and $G_3$ is the smallest of the three graphs for $\mu$.}
	\label{ordre}
\end{figure}

\begin{proof}[Proof of Theorem \ref{extension}]
First, if $\Pi$ can be tested in polynomial time, the extension problem is in NP.
Indeed, given a subset $S$ of vertices of the input graph $G$, one can easily check in polynomial time whether $G[S]$ satisfy $\Pi$, and whether $S$ contains $A$ and avoids $B$.
%we assumed previously that property $\Pi$ can be tested in polynomial time.
Moreover, 
the maximality of $G[S]$ can be tested in polynomial time because $\Pi$ is hereditary.
%The conditions on $A$ and $B$ can also be tested in polynomial time.

As stated earlier, we will assume that the class $\Pi$ contains all the stable sets.
The stable sets are minimum for $\mu$.

The order defined by $\mu$ is well-founded, so every nonempty subset of the set of all graphs has a smallest element for $\mu$.
Hence there exists a graph $\Jpi \notin \Pi$ which is minimum for $\mu$, that is to say $\mu(\Jpi) = \min\limits_{lex} \{ \mu(H)\mid H \notin \Pi \}$.
The graph $\Jpi$ has at least two vertices and at least one edge.

Denote by $J_1,\ldots,J_l$ the connected components of $\Jpi$ (such that $\mu(\Jpi)=(\lambda(J_1),\ldots,\lambda(J_l))$).
Let $J^*_0$ be the largest connected component of $J_1 - c_1(J_1)$, we define as $J_0$ the subgraph of $\Jpi$ induced by $V(J^*_0)\cup\{c_1(J_1)\}$.
Let also $J_1' = \Jpi[V(J_1)\setminus V(J_0^*)]$.
The graphs $J_0$ and $J_1'$ are connected, and $J_0$ has at least one edge.
Let then $d$ be any vertex of $J_0$ different from $c_1(J_1)$.

Before starting the reduction, let us have a look at what happens to the obstruction $\Jpi$. %\lily{On investigue vraiment further la, avant la reduction ? La reduction commence juste apres la figure.}
An illustration of this process applied to a graph $\Jpi$ is given in Figure \ref{obstr}.
This graph verifies $\mu(\Jpi)=((4,2),(3))$. %\lily{Pourquoi ce n'est pas (4,2) pour $J_1$ deja (ici et dans le dessin) ? Je crois qu'on en avait deja parle mais je n'arrive pas a retrouver pourquoi}.

\begin{figure}[!ht]
	\centering
	\includegraphics{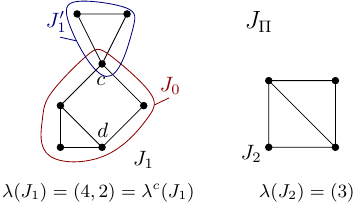}
	\caption{An obstruction $\Jpi$ and its components $J_1$ and $J_2$.}
	\label{obstr}
\end{figure}

\paragraph{}
The reduction is from the maximal stable set avoiding a set $B\subseteq V(G)$, whose NP-completeness is proven in \cite[Proposition 2]{boros}.
Let now $G$ be a graph, and let $B$ be a subset of its vertices.

\paragraph{}
We will build the graph $G'$, using a technique similar to \cite{yannakakis}, as follows.
For each vertex $u$ of $G$, attach a copy of $J_1'$ to $u$ by identifying $c_1(J_1)$ with $u$.

Replace every edge of $G$ by a copy of $J_0$, identifying $c_1(J_1)$ with one endpoint and $d$ with the other, in any order.
Finally, add $J_2,\ldots,J_l$ as new connected components.
Call $G'$ the graph obtained this way.
The graph $G'$ contains a copy of $\Jpi$, therefore it does not verify the hereditary property $\Pi$.
Figure \ref{transf} illustrates the transformation of a graph $G$ into $G'$, using the obstruction $\Jpi$ of figure \ref{obstr}.
Call $A$ the set of all vertices that are in $G'$ but were not in $G$.

\begin{figure}[!ht]
	\centering
	\includegraphics[page=1, width=0.4\textwidth]{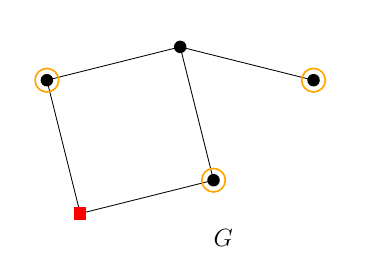}
	\hfill
	\includegraphics[page=2, width=0.4\textwidth]{transformation.pdf}
	\caption{$G$ transforms into $G'$. The set $B$ consists here in the only square vertex, and the set $A$ in the smaller vertices. The circled vertices form a maximal stable set of $G$ not intersecting $B$.}
	\label{transf}
\end{figure}

% CORRIGÉ
We will show that finding a maximal stable set of $G$ avoiding $B$, or finding a maximal induced $\Pi$-subgraph of $G'$ which contains $A$ and avoids $B$ are two equivalent problems.

\paragraph{1)}
First, suppose that there exists a maximal stable set $S$ of $G$ which does not intersect $B$.
We will see that the graph $H:=G'[S\cup A]$ verifies $\mu(H)<\mu(\Jpi)$.
As $S$ is a maximal stable set of $G$, each edge in $G$ has at least one endpoint in $V\setminus S$.
Hence removing from $G'$ all the vertices of $V(G)\setminus S$, destroys a vertex identified to $c_1(J_1)$ or $d$ in each copy of $J_1$.
For each vertex $s\in S$, let us  further observe the connected component $H_s$ of $H$ containing $s$. It is composed of one copy of $J_1'$ and $d_G(s)$ "truncated" copies of $J_0$, all glued together by identifying either $c_1(J_1)$ or $d$ into $s$ (by "truncated" copy of $J_0$, we mean either $J_0 - c_1(J_1)$ or the connected component of $J_0 - d$ that contains  $c_1(J_1)$). We have $\lambda(H_s)\leq\lambda^s(H_s)<_{\text{lex}}\lambda(J_1)$ because the truncated copies of $J_0$ are smaller than $J_0$ itself. Let us show now that $\mu(H)< \mu(J_\Pi)$. Observe that the connected components of $H$ are precisely $\{H_s : s\in S\} \cup \{J_2,...,J_l\}$. We assumed that $\lambda(J_1),...,\lambda(J_l)$ is a non increasing sequence, so $\lambda(J_i) \leq \lambda(J_1)$ for all $2\leq i\leq l$. If $\lambda(J_2)<\lambda(J_1)$ (if the inequality is strict), then since $\lambda(H_s)<\lambda(J_1)$ for all $s\in S$ we conclude that $\lambda(C) < \lambda (J_1)$ for any connected component $C$ of $H$. So we have $\max\{\lambda(C) : C \text{ is a connected component of } H\} < \lambda(J_1)$ and the sequence $\mu(H)$ is lexicographically strictly smaller than $\mu(J_\Pi)$ since the first element of $\mu(H)$ is strictly smaller than the first element of $\mu(J_\Pi)$. Now if $\mu(H)$ was starting with a plateau, i.e. if there exists $j\geq 2$ such that $\lambda(J_1)=\lambda(J_2)=\cdots=\lambda(J_j)$ (choose the largest such index $j$), then
 we will conclude by noticing that the sequence $\mu(H)$ 
 starts with the plateau $\lambda(J_2)=\cdots=\lambda(J_j)$, which is strictly shorter by one, and then all other values appearing in $\mu(H)$ are strictly smaller than $\lambda(J_1)$.
%contains one element of value $\lambda(J_1)$ less than the sequence $\mu(J_\Pi)$. Let $j$ be the largest index such that $\lambda(J_j)=\lambda(J_1)$ i.e. only the first $j$ values of the sequence $\mu(J_\Pi)$ are equals to $\lambda(J_1)$. 
Indeed, since for all $s\in S$, $\lambda(H_s)<\lambda(J_1)$ and since the only other connected components of $H$ are copies of $J_2,...,J_l$, the only connected components $C$ of $H$ for which $\lambda(C)=\lambda(J_1)$ are precisely the copies of $J_2,...,J_j$. So only the first $j-1$ elements of the sequence $\mu(H)$ will have value $\lambda(J_1)$ and then, the sequence  $\mu(H)$ will be lexicographically strictly smaller than $\mu(J_\Pi)$.
In either case, we obtain that $\mu(H)<\mu(J_\Pi)$.

As $\Jpi$ is the smallest obstruction for $\Pi$, and since by definition of $\mu$, $\mu(H)$ is larger than $\mu(K)$ for any subgraph $K$ of $H$, $H$  does not contain any obstruction as an induced subgraph and $H\in \Pi$.
Therefore $H$ is an induced $\Pi$-subgraph of $G'$ that contains $A$ and does not intersect $B$.

% \rouge{It is maximal in $G'$: if it were not, it would be possible to add a new vertex $v\in V(G)$ to $S$, but since $S$ was supposed to be a maximal stable set of $G$, the subgraph the subgraph $G[S \cup \{v\}]$ would contain an edge $uv$. So, in $G'$, this edge whould corresponds to a copy of $J_0$, which together with the copie of $J_1$ attached to $v$ and the copies of $J_2,...,J_l$ would form a subgraph isomorphic to $J_\Pi$ which is an obstruction for $\Pi$.}

\paragraph{2)}
Conversely, assume that there exists a maximal induced $\Pi$-subgraph $H$ of $G'$ containing $A$ and avoiding $B$.
Then $\Jpi$ is not an induced subgraph of $H$.
But every edge $uv$ of $G$ "induces" a copy of $J_1$ in $G'$ (meaning, if you take both $u$, $v$, and a bunch of vertices of $A$) which together with the copies of $J_2,...,J_l$ would form a subgraph isomorphic to $J_\Pi$. Furthermore, $u$ and $v$ are the only vertices of that copy of $J_\Pi$ that do not belong to $A$. Hence, since $J_\Pi$ is an obstruction for $\Pi$ and since $H$ must contain all the vertices of $A$, $H$ does not contain $u$ or does not contain $v$. So $V(H)\cap V(G)$ is a stable set of $G$ which does not intersect $B$.

\smallskip

Up to now we proved :

\begin{itemize}
  \item  In \textbf{1)}  that if $S$ is a maximal stable set of $G$ which does not intersect $B$, then $H:=G'[S\cup A]$ is a $\Pi$-subgraph of $G'$ which contains $A$ and does not intersect $B$.
  \item In \textbf{2)}  that if $H$ is a maximal $\Pi$-subgraph of $G'$ which contains $A$ and does not intersect $B$, then $S:=V(H)\cap V(G)$ is a  stable set of $G$ which does not intersect $B$.
\end{itemize}

It remains to show that what we obtain in both cases is also maximal.

Let us prove that if $S$ is a maximal stable set of $G$ which does not intersect $B$, then $H:=G'[S\cup A]$ is a \textbf{maximal}  $\Pi$-subgraph of $G'$ which contains $A$ and does not intersect $B$.
 Assume it is not maximal and let $H'\supsetneq H$ be a maximal $\Pi$-subgraph of $G'$ which contains $A$ and does not contain $B$. By \textbf{2} we know that $S':=V(H')\cap V(G)$ would be a stable set of $G$ which does not intersect $B$. But then observe that  $S=V(H) \cap V(G) \subseteq V(H') \cap V(G) = S'$. Moreover there exists $v\in V(H')\setminus V(H)$, consequently $v\notin A$ so $v\in V(G)$. Hence $v\in S'\setminus S$ so 
 $S'$ would be a stable set of $G$ that avoids $B$ strictly containing $S$, contradicting the maximality of $S$.

Now let us prove that if $H$ is a maximal $\Pi$-subgraph of $G'$ which contains $A$ and does not intersect $B$, then $S:=V(H)\cap V(G)$ is a \textbf{maximal} stable set of $G$ which does not intersect $B$. Assume that $S$ is not maximal, and let $S'\supsetneq S$ be a maximal stable set of $G$ avoiding $B$ which contains $S$. By \textbf{1)} $H':=G'[S'\cup A]$ is a $\Pi$-subgraph of $G'$ which contains $A$ and does not intersect $B$. But since $H = G'[S\cup A]$, it is contained in $H'$ which would contradict the maximality of $H$ unless $S\cup A=S'\cup A$ which implies $S'=S$ because $A\cap S'=\emptyset$.

% Then  $V(H)\cap V(G)$ is an independent set of $G$ which does not intersect $B$.
% \rouge{This independent set is maximal in $G$: indeed, if there exists $x\in V(G)\setminus A$ such that $(V(H)\cap V(G))\cup\{x\}$ is a stable set of $G$, then $G'[H\cup \{x\}]$ is an induced $\Pi$-subgraph of $G'$ containing $A$, avoiding $B$, and strictly containing $H$.
% Therefore $V(G)\cap V(H)$ is a maximal stable set of $G$ which does not intersect $B$.}

\end{proof}

Hence, for a general class $\Pi$, it is not possible to enumerate all maximal induced $\Pi$-subgraphs of a graph $G$ in polynomial delay using the very simple algorithm derived from the extension problem.
This, however, does not mean that Flashlight search cannot be efficient at all, since sets $A$ and $B$ could in fact be assumed to have certain specific properties depending on $\Pi$, and which could make the extension problem easy to solve in special cases (see for instance \cite{ConteGPU19,KanteLMN14}).

\subsection{Maximal edge-subgraphs}
\label{sec:ext-edge}
The extension problem in its edge version is stated as follows: given a graph property $\Pi$, a graph $G=(V,E)$, and two disjoint subsets $A$ and $B$ of $E$, does there exist a minimal $\Pi$-deletion of $G$ containing all edges of $A$ and no edge of $B$?
As for the vertex version, it does not suffice in general to look for minimal $\Pi$-deletions of $G-B$, because the $\Pi$-deletions obtained this way may not be maximal $\Pi$-subgraphs of $G$.
%\lily{J'ai reformule l'intro en termes de $\Pi$-deletion pour etre raccord avec le debut de l'article, est-ce qu'on traque dans cette section tous les "maximal subgraphs" pour les remplacer par des "minimal deletion" ?}
In the remainder of this section, to draw a parallel between the two versions of the extension problem, minimal completions will be referred to as ``maximal $\Pi$-subgraphs''.

Contrary to the extension problem studied in the previous subsection, this ``edge version'' of the extension problem is not NP-complete for all nontrivial hereditary $\Pi$, as discussed in the following.
%In the case of maximal edge subgraphs, also known as minimal deletions, 
Thus it would be of interest to determine for which classes of graphs the extension problem is polynomial, and for which ones it is NP-complete.

A graph property $\Pi$ is said to be \emph{monotone} if it is closed under edge removal.
In other words, $\Pi$ is monotone if for any graph $G$ satisfying $\Pi$ and any edge $e$ of $G$, $G-e$ also satisfies $\Pi$.

\subsubsection*{Forests}

If $\Pi$ denotes the class of forests, which is a monotone property, the extension problem asks if there exists a maximal spanning forest of $G$ containing $A$ and avoiding $B$.
Assume first that $G$ is connected.
In this particular case, if removing $B$ first disconnects $G$, then the answer is no: all spanning trees of $G$ must contain at least one element from $B$.
Else, $B$ can be removed first.
The set $A$, if it induces a forest, can then be extended in $G-B$ to a spanning tree of $G$.
In the case where $G$ is not connected, a (maximal) spanning forest of $G$ is the union of spanning trees of all connected components of $G$.
Therefore, the extension problem for forests is polynomial.

\bigskip

Nevertheless, the extension problem is not polynomial for all classes of graphs.

\subsubsection*{$P_k$-free graphs}

It is proven in \cite{casel}, among other results, that the extension problem with $A=\varnothing$ is NP-complete in the case of maximal matchings.
Since maximal $P_3$-free graphs are exactly maximal matchings in triangle-free graphs, this result implies that the edge extension problem for $P_3$-free graphs is NP-complete, even in triangle-free graphs and with $A=\varnothing$.
In the sequel, we extend this result, proving the following theorem.

\begin{thm}\label{pk_free}
The edge extension problem for $P_k$-free graphs is NP-hard for all $k\geq 3$.
Moreover, it is NP-complete when $k$ equals $3$ or $4$.
\end{thm}

\begin{proof}%[Proof of Theorem \ref{pk_free}]

For $k=3$, the result follows from \cite{casel} and the previous discussion.
Now, let us prove it for $k=4$.

For $G=(V,E)$ a graph, and a set $S\subseteq E$ of edges of $G$, determining if $S$ is the edge set of a maximal $P_4$-free subgraph of $G$ can be done in polynomial time: it suffices to complete $S$ into a maximal solution $S'$ (iterating the polynomial sandwich algorithm from \cite{sandwich}) and to check if $S=S'$.
The problem is in NP.

\paragraph{}
The reduction is from 3-SAT.
Begin with a 3-SAT formula, with clauses $C_1\ldots,C_m$, and variables $X_1,\ldots,X_n$.
Consider the graph $G$ built from the formula as follows.

For each clause $C_j$, build a vertex $c_j$, and for each variable $X_i$, build two vertices $x_i$ and $\overline{x_i}$ corresponding to the associated literals, linked by an edge $x_i \overline{x_i}$.
Now, add an edge between $x_i$ (or $\overline{x_i}$) and $c_j$ if the corresponding literal appears in clause $C_j$.
At this step, each vertex $c_j$ is of degree 3.
Finally, for all $j$, add a pending vertex $c_j'$ to $c_j$.
Call the resulting graph $G$.
An illustration of such a construction is given in Figure \ref{reductionAretes}.

\begin{figure}[!ht]
	\centering
	\includegraphics{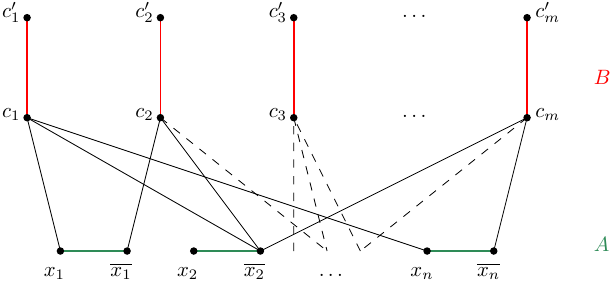}
	\caption{The graph $G$ obtained from a 3-SAT formula.}
	\label{reductionAretes}
\end{figure}

Define the two edge sets $A$ and $B$ by $A=\{x_i \overline{x_i}\mid 1\leq i \leq n\}$ and $B=\{c_j c_j'\mid 1\leq j\leq m\}$.
These will be the sets $A$ and $B$ considered for the extension problem.

%\paragraph{}
Suppose that there exists a maximal $P_4$-free subgraph $H$ of $G$, containing $A$ and avoiding $B$.
Since $H$ is a maximal $P_4$-free subgraph of $G$, adding to $H$ an edge of $B$ produces an induced $P_4$.
That is to say, each vertex $c_j$ is the endpoint of an induced $P_3$.
Thus, there exists $i$ such that $c_j x_i \in E(H)$ or $c_j \overline{x_i} \in E(H)$.
Without loss of generality, suppose that $c_j x_i \in E(H)$ (so $c_j \overline{x_i} \notin E(H)$).
In this case, there is no $k$ such that $c_k \overline{x_i} \in E(H)$, otherwise $c_j x_i \overline{x_i} c_k$ is an induced $P_4$ in $H$.

For all $1\leq i\leq n$, at most one of $x_i$, $\overline{x_i}$ has neighbours among $\{c_1,\ldots,c_m\}$.
Assigning the value $1$ to the corresponding literal satisfies the formula.

%\paragraph{}
Conversely, if the formula is satisfiable, there exists a maximal $P_4$-free subgraph $H$ of $G$ containing $A$ and avoiding $B$.
It suffices to keep only edges incident to vertices corresponding to literals assigned the value $1$.
To avoid creating induced $P_4$, keep only one edge incident to each vertex $c_j$.
This way, each vertex $c_j$ is the endpoint of an induced $P_3$, and no edge of $B$ can be added.

\paragraph{}
Consequently, the edge extension problem for $P_4$-free graphs is NP-complete, even in triangle-free graphs because the graph $G$ we built is triangle-free.

It is easy to adapt the previous reduction to show the NP-hardness of the problem for $P_k$-free graphs, with $k>4$.
It suffices to subdivide each edge $x_i \overline{x_i}$ into a path of length $k-3$ whose edges are all put in the set $A$.
\end{proof}

\subsubsection*{Graphs without cycles of length at most $k$}

Another class of graphs that is worth studying is the class of graphs without cycles of length at most $k$, for a given $k\geq 3$, also known as graphs of girth at least $k+1$ (the \emph{girth} of a graph is the length of its shortest cycle).
This is a monotone property.
For this class of graph, it has been shown by Yannakakis in \cite{yannakakis2} that the \emph{edge-deletion problem} is NP-complete.
In this case, the edge-deletion problem consists in finding a set of edges of minimum cardinality whose removal results in a graph without cycles of length at most $k$.
In the same fashion, we prove that the edge extension problem for graphs without cycles of length at most $k$, that is to say, determining if there exists a maximal edge subgraph without cycles of length at most $k$, containing $A$ and avoiding $B$, is also NP-complete.

\begin{thm}\label{ck_free}
The edge extension problem for graphs without cycles of length at most $k$ is NP-complete for all $k\geq 3$.
\end{thm}

%\begin{proof}[Proof of Theorem \ref{ck_free}]
\begin{proof}
First of all, checking if a set of edges induces a subgraph without cycles of length at most $k$ can be done in time $\bigo(n^k)$, where $n$ is the number of vertices of the input graph: it suffices to test all possible vertex subsets of size at most $k$.
Moreover, as the property is monotone, maximality can also be checked in polynomial time.

\paragraph{}
The reduction is from the minimal Vertex Cover extension in triangle-free graphs, whose NP-complete\-ness is proven in \cite[Proposition 2]{boros}.
The minimal Vertex Cover extension problem asks, given a graph $G$ and a subset $S$ of its vertex set, whether there exists a minimal Vertex Cover of $G$ containing $S$.
The NP-completeness statement in \cite{boros} does not mention the triangle-free case, but the reduction used in the proof produces only triangle-free input graphs, which proves the NP-completeness we need.
%Since the proof only uses triangle-free graphs, this problem is still NP-complete when restricted to triangle-free graphs.

\paragraph{}
First, we shall prove the result for $k=3$.
Let $G=(V,E)$ be a triangle-free graph, let $S\subseteq V$.
Then we build another graph $G+c$ by adding a universal vertex $c$ to $G$.
Since $G$ is triangle-free, all triangles in $G+c$ must include $c$.
Moreover, for two vertices $u,~v\in V$, the three vertices $u$, $v$, and $c$ form a triangle in $G+c$ if and only if $uv\in E$.
Thus, there is a one-to-one correspondence between the edges of $G$ and the triangles of $G+c$.
Finally, define $A$ and $B$ two subsets of edges of $G+c$ as follows: $A:=E(G)$ and $B:=\{cv\mid v\in S\}$.
An illustration of this transformation is presented in Figure \ref{extAreteK3}.

\begin{figure}[ht]
	\centering
	\includegraphics[page=1]{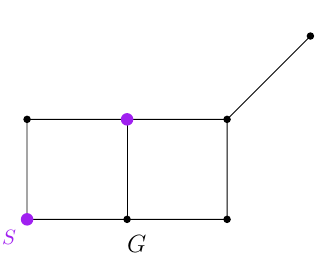}
	\hspace{1cm}
	\includegraphics[page=2]{extAreteK3.pdf}
	\caption{Transformation of an instance of Vertex Cover extension into an instance of the triangle-free edge extension problem. Big vertices are vertices in $S$, $A$ is the set of green edges, and $B$ is the set of red edges.}
	\label{extAreteK3}
\end{figure}

We will show that finding an edge extension without cycles of length at most 3 is equivalent to find a minimal Vertex Cover of $G$ containing $S$, or to find a maximal triangle-free subgraph of $G+c$ containing $A$ and avoiding $B$.

\paragraph{}
Let $H$ be a maximal triangle-free subgraph of $G+c$, containing $A$ and avoiding $B$ and let $X=\{v\in V \mid cv\notin E(H)\}$. Let us prove that $X$ is a minimal Vertex Cover of $G$ containing $S$.
Since $H$ is triangle-free, $X$ is a Vertex Cover of $G$.
Indeed, for each edge $e$ of $G$ there exists $v\in e$ such that $cv\notin E(H)$ (hence $v\in X$), otherwise $e\cup\{c\}$ induces a triangle in $H$.

Moreover, since $H$ does not contain any edge of $B$, $X$ contains $S$ for otherwise there would be $v\in S$ such that $cv\in E(H)$.

Finally, since $H$ is a maximal triangle-free subgraph of $G+c$, it is easy to see that the Vertex Cover $X$ is minimal in $G$: if there were $v \in X$ such that $X\setminus\{v\}$ is a Vertex Cover of $G$, then all neighbours of $v$ (except $c$) would have to be in $X$ to cover all the edges.
So, by definition of $X$, there is no neighbour $u$ of $v$ such that $cu\in E(H)$, hence we can add the edge $cv$ to $H$ without creating any triangle.
This contradicts the maximality of $H$.
\paragraph{}
Conversely, if we are given a minimal Vertex Cover of $G$ containing $S$, then the subgraph of $G+c$ given by $(V\cup\{c\},A\cup\{cv \mid v\notin S\})$ is triangle-free, and it is maximal in $G+c$ because the Vertex Cover is minimal in $G$.

This concludes the proof for $k=3$.

\paragraph{}
For greater values of $k$, the proof can easily be adapted by subdividing each edge of $G$ in $G+c$ into a path of length $k-2$ (that is, adding $k-3$ new vertices on each edge), and defining $A$ to contain all edges of those paths.
Such a transformation is illustrated in Figure \ref{extAreteK3-2}.
%\lily{Caroline, pour info, il faut mettre ton label apres la caption pour que des references soient justes. La ca indiquait  Fig 8.2}
Note that for $k>3$, the smallest cycle in the subgraph induced by $A$ has length at least $3(k-2) > k$, ensuring that the solutions of the extension problem are not trivially inexistent.

\begin{figure}[ht]
	\centering
	\includegraphics[page=3]{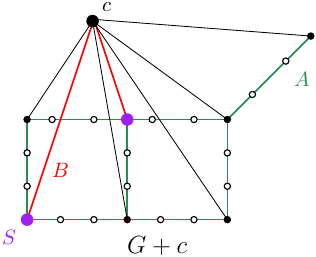}
	\caption{Transformation for $k=5$.}
	\label{extAreteK3-2}
\end{figure}
\end{proof}

\bibliographystyle{plain}
\bibliography{biblio}

\end{document}